\documentclass[12pt, onecolumn]{IEEEtran}

\usepackage{amsmath}
\usepackage{amsthm}
\usepackage{cite}
\usepackage{pb-diagram}
\usepackage{amssymb}
\usepackage{mathrsfs}
\usepackage{array}
\usepackage{enumerate}
\usepackage{color}
\usepackage{graphicx}
\usepackage{mathrsfs}
\usepackage{url}
\usepackage{stmaryrd} 

\usepackage{fancyhdr}

\newcommand{\Al}{|S|}

\newcommand{\QQ}{\mathbb{Q}}

\newcommand{\Ac}{\mathcal{A}}
\newcommand{\Cc}{\mathcal{C}}

\renewcommand{\mod}{\mbox{mod }}

\newcommand{\R}{\mathbb{R}}
\newcommand{\C}{\mathbb{C}}

\newcommand{\Q}{\mathbb{Q}}

\newcommand{\ds}{\displaystyle}

\newcommand{\qa}[3]{({#1},{#2})_{#3}}

\newcommand{\mm}{\mathfrak{m}}
\def\ov#1{{\overline{#1}}}
\newcommand{\func}[4]{\begin{aligned}\newline #1&\longrightarrow #2 \cr \newline #3 &\longmapsto #4\end{aligned}}
\newcommand{\Nrd}{\mathrm{Nrd}}
\newcommand{\rr}{\mathbb{R}}

\newcommand{\SH}{skew-Hermitian}
\newcommand{\MO}{mutually orthogonal}

\newcommand{\mbfs}{\mathbf{s} }
\newcommand{\mbfv}{\mathbf{v} }
\newcommand{\mbfw}{\mathbf{w} }

\newcommand{\vecC}{\text{Vec}_\C}
\newcommand{\vecR}{\text{Vec}_\R}

\newcommand{\marg[1]}{ \marginpar{\textcolor{red}{#1}}}

\newtheorem{proposition}{Proposition}
\newtheorem{lemma}[proposition]{Lemma}
\newtheorem{corollary}[proposition]{Corollary}

\newtheorem{theorem}{Theorem}
\newtheorem{defi}{Definition}

\theoremstyle{remark}
\newtheorem{remark}{Remark}
\newtheorem{example}{Example}

\begin{document}

\title{Bounds on fast decodability of space-time block codes, skew-Hermitian matrices, and Azumaya algebras}  
\date\today

\author{Gr\'egory Berhuy, Nadya Markin, and B. A. Sethuraman
\thanks{Gr\'egory Berhuy is with Universit\'e Joseph Fourier, Institut Fourier, 100 rue des maths, BP 74, F-38402 Saint Martin d'H\'eres Cedex, France.
E-mail: Gregory.Berhuy@ujf-grenoble.fr} 
\thanks{Nadya Markin is with Division of Mathematical Sciences, School of Physical and Mathematical Sciences, Nanyang Technological University, Singapore.
E-mail: NMarkin@ntu.edu} 
\thanks{B.A. Sethuraman is with Department of Mathematics, California State University Northridge,  Northridge, CA 91330, USA. 
E-mail: al.sethuraman@csun.edu} 

}

\maketitle

\begin{abstract} We study fast lattice decodability of space-time block codes for $n$ transmit and receive antennas, written very generally as a linear combination $\sum_{i=1}^{2l} s_i A_i$, where the $s_i$ are real information symbols and the $A_i$ are $n\times n$ $\R$-linearly independent complex valued matrices.  We show that 
the mutual orthogonality condition $A_iA_j^* + A_jA_i^*=0$ for distinct basis matrices is not only sufficient but also necessary for
fast decodability.  We build on this to show that for full-rate ($l = n^2$) transmission, the decoding complexity can be no better than $\Al^{n^2+1}$, where $\Al$ is the size of the effective real signal constellation. We also show that for full-rate transmission, $g$-group decodability, as defined in \cite {JR}, is impossible for any $g\ge 2$. We then use the theory of Azumaya algebras to derive bounds on the maximum number of groups into which the basis matrices can be partitioned so that the matrices in different groups are mutually orthogonal---a key measure of fast decodability.  We show that in general, this maximum number is of the order of only the $2$-adic value of $n$.  In  the case where the matrices $A_i$ arise from a division algebra, which is most desirable for diversity, we show that the maximum number of groups is only $4$. As a result, the decoding complexity for this case is no better than $\Al^{\lceil l/2 \rceil}$ for any rate $l$.
\end{abstract}

\begin{keywords} 
Fast Decodability, Full Diversity, Full Rate, Space-Time Code, Division Algebra, Azumaya Algebra.
\end{keywords}

%
%
\section{Introduction}\label{introsection}

Space-time block codes for multiple input multiple output communications with $n$ transmit and receive antennas and delay $n$ and where the channel is known to the receiver consist of $n\times n$ matrices $X=X(x_1, \dots, x_l)$, $l \le n^2$, where the symbols $x_i$ arise from a finite subset $S$ of the nonzero complex numbers.  The matrices are generally assumed to be linear in the $x_i$, so splitting each $x_i$ into its real and imaginary parts, we may write $X = \displaystyle\sum_{i=1}^{2l} s_i A_i$, where the $s_i$ are real valued drawn from the effective real signal constellation $S$, and the $A_i$ are fixed $\R$-linearly independent complex valued matrices.  The transmission process may then be modeled as one where points from a $2l$-dimensional lattice in $\R^{2n^2}$ are transmitted (with the lattice changing every time the channel parameters change), and the decoding modeled as a closest lattice-point search.

Since closest lattice-point searches are notoriously difficult in general (although approximate  decoding methods like sphere decoding  \cite{VB} exist, which, by restricting the  search points to a small region around the received point, speed up the process in small dimensions), much attention has been paid lately on selecting the matrices $A_i$ above so that the resulting lattice breaks off as nearly as possible into an orthogonal direct sum of smaller dimensional lattices generated by some subsets of the canonical basis vectors, \textit{no matter what the channel parameters} (see Remark \ref{rem_fd_sublattice} ahead for the interpretation of the previous work in terms of orthogonal sublattices).   
This then reduces the complexity of decoding from the worst case complexity $\Al^{2l}$ which arises from a brute-force checking of all ${2l}$-tuples from $S$, 
to the order of $\Al^{l'}$ for some $l' < 2l$,  where $l'$ depends on the dimensions of the orthogonal summands. Some examples of recent work on fast decoding include  \cite{BHV}, \cite {RenEtAl}, \cite{JR}, \cite{SR}, \cite{MO}, \cite{OVH}, \cite{LO}, \cite{SRGAg}, \cite{NR}. 
Many codes have been shown to have reduced decoding complexity; for instance, it is known that the Silver code has a  decoding complexity that is no higher than $\Al^5$ (instead of the possible $\Al^8$) \cite[Example 5]{JR}, considered in Example \ref{ex:bound_is_strict} ahead. 

By \textit{decoding complexity} we will mean throughout the complexity of the worst case decoding process whereby, upon possibly conditioning some variables, a brute-force check of the decoding metric is performed for all tuples from the remaining variables, possibly in parallel if the lattice has orthogonal direct summands.  This is to be contrasted with other decoding processes that may exist that avoid brute force checking of the metric for all tuples, such as the GDL decoder described in \cite{NR3}.

In this paper, we analyze the conditions on the basis matrices $A_i$ needed for reduced decoding complexity of space-time block codes arising from the phenomenon described above: the presence of orthogonal direct sums of smaller dimensional lattices generated by some subsets of the basis vectors of the transmitted lattice, no matter what the channel parameters.  
We show that the condition $A_iA_j^* + A_jA_i^*=0$ for various distinct basis matrices $A_i$ and $A_j$, previously considered in the literature primarily as a sufficient condition (\cite{JR} or \cite{SR} for instance, see also \cite{RenEtAl}), is actually a \textit{necessary} condition (although, this result had indeed been proven before \cite{YuenEtAl} using different techniques than ours, a fact we were unaware of: see Remark \ref{sufficiency_only}  ahead as well).    We analyze this condition further, using just some elementary facts about skew-Hermitian and Hermitian matrices, and show that for a full-rate code (i.e., where $l=n^2$), the decoding complexity cannot be improved below $\Al^{n^2+1}$. We also show that for a full-rate code, the transmitted lattice cannot be decomposed entirely as an orthogonal direct sum of smaller dimensional lattices generated by the  basis vectors (a condition referred to as $g$-group decodability by previous authors, for instance \cite{JR}.) 

We then drop the assumption of full rate and turn to the maximum number of orthogonal sublattices generated by basis vectors that is possible in the transmitted lattice;
the dimension of the various sublattices then controls the fast-decodability.  We use the theory of Azumaya algebras to show that the number of such summands is bounded above by $2 v_2(n)+4$  in general (where $v_2(n)$ is the $2$-adic value of $n$, i.e., the highest power of $2$ in the prime factorization of $n$).  In the process, we generalize the classical Radon-Hurwitz-Eckmann bound \cite{Ecm} on the number of unitary matrices of square $-1$ that skew commute.
Our method allows us to consider not just the general case but the special cases where the matrices $A_i$ arise from embeddings of matrices over division algebras, where the bound on the number of summands becomes even smaller.  In the case where the $A_i$ come from the embedding of a division algebra, which is of most interest since codes from division algebras satisfy the full diversity criterion,  we show that the maximum number of possible summands is very low: just $4$ in fact.  This then shows that the decoding complexity of a code arising from a division algebra cannot be made better than $\Al^{\lceil l/2 \rceil}$.

The paper is organized as follows: After some preliminary background on vectorizations of matrices and on Hermitian and skew-Hermitian matrices in Section \ref {secn_prelim}, we describe the system model and maximum likelihood decoding in  Section \ref{secn_sys_model}. We then discuss fast decodability in Section \ref{secn_fast_decod} and derive the equivalence of fast decodability to the mutual orthogonality of subsets of the basis matrices. In Section \ref{secn_mo_shmo} we analyze the mutual orthogonality condition using properties of skew-Hermitian and Hermitian matrices, and derive our lower bounds on the decoding complexity of full-rate codes. In Section \ref{AzAlg}, we use the theory of Azumaya Algebras to derive the bound on the number of orthogonal sublattices generated by basis vectors.  Necessary background from commutative algebra and Azumaya algebras is collected in the appendices.

\par\noindent\textit{Acknowledgements:} N. Markin was supported by the Singapore National Research Foundation under Research Grant NRF-RF2009-07. B.A. Sethuraman was supported by a U.S. National Science Foundation grant CCF-1318260. G. Berhuy and B.A. Sethuraman wish to thank Prof. Frederique Oggier and Nanyang Technological University, Singapore, for hosting their visit during which the ideas for this paper germinated. Portions of this paper were presented at the ISIT 2014 conference \cite{ISIT}.

%
%
\section{Preliminaries}
\label {secn_prelim}

For any vector $\mbfv \in \C^n$, we let 
$$\vecR(\mbfv) = (Re(v_1), Im(v_1), \ldots, Re(v_n), Im(v_n))^t$$ be the vector in $\R^{2n}$ whose $2i-1^{th}$ coordinate is the real part of $\mbfv_i$ and whose $2i$-th coordinate is the imaginary part of $\mbfv_i$.  For any matrix $A\in M_n(\C)$, we will write $\vecC(A)$ for the vector in $\C^{n^2}$ obtained by stacking the entries of $A$ in some fixed order (e.g. column $1$ then column $2$, etc.).  To simplify notation, for a matrix $A$ in $M_n(\C)$, we will directly write $\vecR(A)$ for the vector $\vecR(\vecC(A))$ in $\R^{2n^2}$.

For two vectors $\mbfv$ and $\mbfw$ in $\C^n$, we write $\langle\mbfv,\mbfw\rangle_\C$ for the usual Hermitian product in $\C^n$, namely, $\langle\mbfv,\mbfw\rangle_\C = \mbfv \cdot \mbfw^* = \mbfv\cdot \overline{\mbfw}^t$ (where the superscript $t$ stands for transpose).  For two vectors $\mbfv$ and $\mbfw$ in $\R^n$, $\mbfv \cdot \mbfw$ will denote the dot product of the two vectors. For any matrix $A \in M_n(\C)$, we will write $A^*$ for the conjugate transpose of $A$, i.e., $A^* = \overline{A}^t$.
Also, we will write $\text{Tr}$ for the trace of a matrix, $\text{Re}$ for the real part of a complex number.

The following are elementary:

\begin{lemma} \label{Lemma_TrAB} For two matrices $A$ and $B$ in $M_n(\C)$, $\langle \vecC(A), \vecC(B)\rangle_\C = \text{Tr}(AB^*)$.
\end{lemma}

\begin{lemma} \label{Lemma_RealDotProd} For two vectors $\mbfv$ and $\mbfw$ in $\C^n$, $\vecR(\mbfv) \cdot \vecR(\mbfw) = \text{Re}(\langle\mbfv,\mbfw\rangle_\C)$.
\end{lemma}

We immediately get the following corollary: 

\begin{corollary} \label{vec_mat_dot_prod} For two matrices $A$ and $B$ in $M_n(\C)$, we have $\vecR(A) \cdot \vecR(B) = \text{Re}\left(\text{Tr}(AB^*) \right)$. In particular, for matrices $A$ and $B$, $\vecR(A)$ and $\vecR(B)$ are orthogonal in $\R^{2n^2}$ if and only if $\text{Re}\left(\text{Tr}(AB^*) \right)=0$.

\end{corollary}

We recall that a matrix $A \in M_n(\C)$ is Hermitian if $A^* = A$, and skew-Hermitian if $A^* = -A$. The matrix $\imath I_n$ (where $\imath$ is a square root of $-1$ and $I_n$ is the identity $n\times n$ matrix) is \SH.  The set  $H_n$ of all Hermitian matrices and the set $SH_n$ of all \SH \ matrices in $M_n(\C)$ each forms a vector space over $\R$, each of dimension $n^2$.  Moreover, for any Hermitian matrix $A$, $\imath A$ is \SH, and for every \SH \ matrix $B$, $\imath B$ is Hermitian. Every matrix can be written uniquely as a sum of a Hermitian and a \SH \ matrix, i.e., $M_n(\C) \cong H_n \oplus SH_n$ as $\R$-vector spaces. We will need to use these facts in the paper.

%
%

\section{System Model and Maximum Likelihood Decoding }\label{secn_sys_model}
We consider transmission over a quasi-static Rayleigh fading channel with perfect channel state information at the receiver. We assume that the number of receive antennas and the number of transmit antennas are the same, namely $n$, and we assume the the block length, i.e., the number of times we transmit through the channel before processing, is also $n$. The codewords are $n\times n$ complex valued matrices
$X=X(x_1, \dots, x_l)$, $l \le n^2$, where the symbols $x_i$ arise from a finite subset of the nonzero complex numbers.  The matrices $X$ are assumed to be linear in the $x_i$, so splitting each $x_i$ into its real and imaginary parts, we may write $X = \displaystyle\sum_{i=1}^{2l} s_i A_i$, where the $s_i$ are real symbols arising from the effective real alphabet $S$, and the $A_i$ are fixed $\R$-linearly independent complex valued matrices. We will assume throughout the paper that the $A_i$ are invertible, which is not a significant constraint, since  invertible matrices form a dense subset of $n \times n$ complex matrices; besides, when the space-time code is fully diverse (which is the desirable situation), the matrices $A_i$ are necessarily constrained to be invertible.

The received codeword is given by \begin{equation} \label {system_model}
Y = HX + N
\end{equation}
where $H \in M_n(\C)$ is the channel matrix and $N \in M_n(\C)$ is the noise matrix.  It is assumed that the entries of $H$ are i.i.d. circularly symmetric complex Gaussian with zero mean and variance $1$, and the entries of $N$ are i.i.d. complex Gaussian with zero mean and variance $N_0$.

The statistics of $N$ shows that Maximum-likelihood (ML) decoding amounts to finding the information symbols $s_1, \dots, s_{2l}$ that result in a codeword 
$X =  \displaystyle\sum_{i=1}^{2l} s_i A_i$  which minimizes the squared Frobenius norm
\begin{equation}
\label{frob-min}
||Y-HX||_F^2.
\end{equation}

The transmission process may be modeled as  one where points from a $2l$-dimensional lattice in $\R^{2n^2}$ are transmitted, with the lattice changing every time the channel matrix $H$ changes, and the decoding modeled as a closest lattice-point search in $\R^{2n^2}$. We do this  as follows: We convert the matrices appearing in Equation \ref{system_model}  to vectors in complex space and then further split the complex entries into their real and imaginary parts: 
$$
\vecR(  Y ) = \displaystyle\sum_{i=1}^{2l} s_i \vecR(HA_i) + \vecR(N).
$$

We define $T=T(H)$ to be the $2n^2 \times {2l}$ matrix over $\R$ whose $i$-th column is $\vecR(HA_i)$. 
Then we have 
$$ 
 \displaystyle\sum_{i=1}^{2l} s_i \vecR(HA_i)= T(s_1, \dots, s_{2l})^t $$

Thus, $T = T(H)$ is the basis matrix for the $2l$-dimensional lattice in $\R^{2n^2}$ from which points are  transmitted. Writing $\mbfs$ for the vector $(s_1, \dots, s_{2l})^t$, the decoding problem now becomes to 
find a maximum likelihood estimate for the symbols $s_1$, $\dots$, $s_{2l}$ from the linear system of equations in $\R^{2n^2}$
\begin{equation} \label{basic_system}
\vecR(  Y ) = T\cdot \mbfs+ \vecR(N), 
\end{equation} where the entries of $\vecR(N)$ are i.i.d. real Gaussian. In other words, the decoding problem is to find an information vector $\mbfs=(s_1, \dots, s_{2l})^t$ which minimizes the Euclidean distance  
\begin{equation} \label{basic_system_min}
|\vecR(  Y ) - T\mbfs|
\end{equation}
of vectors in $\R^{2n^2}$.

Note that the transmitted lattice matrix $T=T(H)$ in Equation \ref{basic_system} above depends on the channel matrix $H$.  

%
%

\section{Fast Lattice Decodability} \label{secn_fast_decod}

Several authors (\cite{BHV}, \cite{JR}) studied fast lattice decodability of space-time codes by considering a $QR$ decomposition of the transmitted lattice matrix $T$  in Equation \ref{basic_system} above (as in the sphere decoder), and rewriting Equation \ref{basic_system} as
\begin{equation}
Q^* \vecR(Y) = R\cdot \mbfs+ Q^*\vecR(N).
\end{equation}
Since $Q^*$ is unitary, the new noise vector $ Q^*\vecR(N)$ is still i.i.d. real Gaussian, so the maximum likelihood estimate for $\mbfs$ is given by minimizing $|Q^* \vecR(Y) -R\cdot \mbfs|$. Fast lattice decodability as defined in \cite{BHV}, \cite{JR} involves choosing the basis matrices $A_i$ so that for all $H$, the matrix $R$ (which depends on $T(H)$ and hence on $H$), has zeros in certain convenient places (see Equation (\ref{fast_decod_matrix})
 ahead in the statement of Theorem \ref{fd_means_nice_T}, for instance). These places are such that  decoding can proceed, after fixing certain $s_i$ if necessary, as parallel decoding of smaller sets of variables, enabling thereby a reduction in complexity.  We will study this process in this section, 
and prove the main result that enables us in the remaining sections to analyze bounds on fast decodability: the equivalence of fast decodability to mutual orthogonality of subsets of the basis matrices $A_i$ (Theorem \ref{thm:FDMOequiv}).

\begin{defi} \label{defn_MO}
We say that two complex matrices, $A, B$ are {\emph{mutually orthogonal}} if 
$AB^*+BA^*=0.$
\end{defi}

We chose this term because, as we  show in  Theorem \ref{orthog_means} below, two basis matrices $A_i$ and $A_j$ satisfy the relation $A_iA_j^* + A_jA_i^* = 0$ if and only if the $i$-th and $j$-th columns of $T$ are mutually orthogonal as vectors in $\R^{2l}$. (Although our proof is new, see Remark \ref{sufficiency_only}  ahead.) The following lemma shows that mutually orthogonal matrices are necessarily $\R$-linearly independent:

\begin{lemma} \label {lem:mo_implies_li}
If $A_1$, $\dots$, $A_n$ are pairwise mutually orthogonal invertible matrices in $M_n(\C)$, then they are $\R$-linearly independent.
\end{lemma}

\begin {proof} Assume that $r_1 A_1 + \cdots + r_n A_n =0$. Multiplying this equation on the right by $A_i^*$, and multiplying the conjugate transpose form of this equation on the left by $A_i$, and then adding, we find $2 r_i A_iA_i^* = 0$.  Since the $A_i$ are invertible, we find $r_i = 0$.
\end{proof}

\begin{theorem} \label{orthog_means} The $i$-th and $j$-th columns of $T = T(H)$ are orthogonal as vectors in $\R^{2l}$ for all channel matrices $H$  if and only 
the basis matrices $A_i$ satisfy 
$A_iA_j^* + A_jA_i^* = 0$.\end{theorem} 
\begin{proof}
We have already noted (Corollary \ref{vec_mat_dot_prod} applied to the definition of the matrix $T$) that the orthogonality of the $i$-th and $j$-th columns of $T$ 
is equivalent to the condition $\text{Re}\left(\text{Tr}((HA_i)(HA_j)^*) \right) = 0$. Also, 
note that $\text{Tr}((HA_i)(HA_j)^*) = \text{Tr}(HA_i A_j^* H^*) = \text{Tr}((A_iA_j^*)(H^*H))$, where the second equality is because $\text{Tr}(XY) = \text{Tr}(YX)$ for two matrices $X$ and $Y$.  

Now assume that $A_iA_j^* + A_jA_i^* = 0$ for $i\neq j$.  Then $A_iA_j^*$ is skew-Hermitian, while $H^*H$ is of course Hermitian.  If $M$ is skew-Hermitian and $P$ is Hermitian, then note that $(MP)^* = P^*M^* = -PM$.  Since for any matrix $X$ we have $\text{Re}(\text{Tr}(X)) = \text{Re}(\text{Tr}(X^*))$, we find that for $X=MP$,  $\text{Re}(\text{Tr}(MP)) = \text{Re}(\text{Tr}((MP)^*)) = \text{Re}(\text{Tr}(-PM)) = -\text{Re}(\text{Tr}(PM)) = -\text{Re}(\text{Tr}(MP))$.  It follows that $\text{Re}(\text{Tr}(MP)) = 0$.  In particular, for $M=A_iA_j^*$ and $P = H^*H$, we find $0=\text{Re}(\text{Tr}(A_iA_j^*)(H^*H)) = \text{Re}(\text{Tr}(HA_i)(A_j^*H^*)) =\text{Re}(\text{Tr}(HA_i)(HA_j)^*)$.

Now assume that the trace condition holds. We write this as $\text{Re}\left(\text{Tr}((A_iA_j^*)(H^*H))\right) = 0$ for all matrices $H$. Write $M$ for $A_iA_j^*$.  We wish to show that $M$ is skew-Hermitian. The matrix $E_{k,k}$ that has $1$ in the $(k,k)$ slot and zeros elsewhere satisfies $E_{k,k}^* E_{k,k} = E_{k,k}$. Choosing $H = E_{k,k}$, we find that the matrix $MH^*H = ME_{k,k}$ will have the $k$-th column of $M$ in the $k$-th column, and zeros elsewhere.  The trace condition now shows that the $(k,k)$ element of $M$ is purely imaginary.
We next need to show that $m_{l,k} = - \overline{m_{k,l}}$ for $k\neq l$, where we have written $m_{i,j}$ for the $(i,j)$-th entry of $M$.
Computing directly, we find the following relations hold (where $E_{i,j}$ has $1$ in the $(i,j)$ slot and zeros everywhere else):

\begin{eqnarray*}
E_{k,k} + E_{k,l} + E_{l,k} + E_{l,l} &=& (E_{k,k} + E_{l,k})\cdot (E_{k,k} + E_{k,l}) \\
E_{k,k} - \imath E_{k,l} + \imath E_{l,k} + E_{l,l} &=& (E_{k,k} + \imath E_{l,k})\cdot (E_{k,k} -\imath E_{k,l}) \\
\end{eqnarray*}
Thus, each of the matrices on the left sides of the two equations above can be written as $H^*H$ for suitable matrices $H$.
Again computing directly, we find  that $M\cdot(E_{k,k} + E_{k,l} + E_{l,k} + E_{l,l})$ has 
$m_{k,k} + m_{k,l}$ in the $(k,k)$ slot and $m_{l,k} + m_{l,l}$ in the $(l,l)$ slot, and zeros elsewhere in the diagonal. Hence, 
 $\text{Re}(\text{Tr} (M\cdot(E_{k,k} + E_{k,l} + E_{l,k} + E_{l,l}))) = \text{Re}( m_{k,k} + m_{k,l} + m_{l,k} + m_{l,l} )$.  Since we have already seen that the diagonal elements of $M$ are purely imaginary, we find $\text{Re}(m_{k,l} + m_{l,k}) = 0$. Similarly, we find 
$\text{Re}(\text{Tr} (M\cdot(E_{k,k} - \imath E_{k,l} + \imath E_{l,k} + E_{l,l} ))) = \text{Re}( m_{k,k} + \imath m_{k,l} -\imath m_{l,k} + m_{l,l} )$. Once again, because the diagonal elements of $M$ are purely imaginary, we find $\text{Im}(m_{k,l} -m_{l,k}) = 0$. These two together show that $m_{l,k} = - \overline{m_{k,l}}$ for $k\neq l$.  Together with the fact that the diagonal elements of $M $ are purely imaginary, we find $M = A_iA_j^*$ is skew-Hermitian, as desired.

\end{proof}

\begin{remark} \label{sufficiency_only} As mentioned in Section \ref{introsection}, the sufficiency of the condition $A_iA_j^* + A_jA_i^* = 0$ for orthogonality of the columns of $T$ and hence for fast decodability was already considered before (\cite[Theorem 2]{SR}, \cite[Theorem 1]{RenEtAl}). What is new here is the necessity of the condition. It is the consequences of the necessity that enables us to analyze lower bounds on fast decodability in the sections ahead by studying the consequences of the condition $A_iA_j^* + A_jA_i^* = 0$.  We should remark, however, that we noticed 
after we proved our results, that the authors of the paper \cite{RenEtAl} also mention the necessity of this condition.  However, they do not give a proof of the necessity in that paper. Tracking this further, we discovered that the authors of \cite{YuenEtAl} have actually provided a proof of this result.  Their proof is  by an explicit computation. Indeed, they write down the entries of $T(H)$, blockwise, in terms of the matrices $H$ and $A_i$, and compute $T(H)^*T(H)$. From the derived block structure of  $T(H)^*T(H)$ they read off the necessity of the mutual orthogonality. This is of course very different from our approach.

\end{remark}

The theorem above allows us to define fast-decodability of a code in terms of its generating matrices, independently of the channel matrix $H$.

\begin{defi}\label {deffd} [See  e.g., \cite[Definition 5]{JR}] 
We will say that the space-time  block code defined by the matrices $X = \sum_{i=1}^{2l} s_i A_i$ admits fast (lattice) decodability if for $g \geq 2$ there exist disjoint 
subsets $\Gamma_1$, $\dots$, $\Gamma_g, \Gamma_{g+1}$, with $\Gamma_{g+1}$ possibly empty, of cardinalities $n_1$, $\dots$, $n_g$, $n_{g+1}$ respectively, whose union is $\{1,\ldots, 2l\}$, 
such that for all $u\in \Gamma_i$ and $v\in \Gamma_j$ ($1 \le i < j \le g $), the generating matrices $A_u, A_v$ are mutually orthogonal. 

\end{defi}

\begin{remark} 
Given a code that admits fast (lattice) decodability, 
 we can define a permutation  
$$\pi:  \{1, \dots, 2l\}\rightarrow \Gamma_1 \cup \ldots \cup \Gamma_g \cup \Gamma_{g+1} ,$$ 
which sends 
 the first $n_1$ elements $\{1, \ldots, n_1\}$ to $\Gamma_1$,
 the next $n_2$ elements $\{n_1+1, \ldots, n_1+n_2\}$ to $\Gamma_2$ 
 and so on, 
where, as in Defintion \ref{deffd}, $n_i  = |\Gamma_i|$ for $i = 1, \ldots, g+1$. Given such permutation $\pi$, we write $T_\pi $ (or $ T_\pi(H) $ for emphasized  dependence on $H$) for the matrix 
{whose $i$-th column is the $\pi(i)$-th column of $T(H)$, namely,} 
$
\vecR(HA_{\pi(i)})$. Similarly, given the vector  $\mbfs = (s_1, \dots, s_{2l})^t$, we write $\mbfs_\pi$ for the vector whose $i$-th component is the $\pi(i)$-th component of $\mbfs$.
\label{rem:deffd}.\end{remark}

{We are now able to link Definition \ref{deffd} of fast-decodability to that given in \cite[Definition 4]{JR}. While the latter definition invokes the channel matrix $H$, the two definitions are actually equivalent, for we have the following result: 
}
\begin{theorem} \label{fd_means_nice_T}  The space-time block  code $X = \sum_{i=1}^{2l} s_i A_i$ admits fast (lattice) decodability as per Definition \ref{deffd} if and only if  there exists a permutation $\pi$ of the index set  $\{1, \dots, 2l\}$, integers $g\ge 2$, $n_i \ge 1$ ($i=1, \dots, g$), and $n_{g+1} \ge 0$, with  $n_1 + \cdots + n_{g+1} = 2l$, such that for all channel matrices $H$, the matrix $R$ obtained by doing a $QR$ decomposition on $T_\pi = T_\pi(H)$ by doing a Gram-Schmidt orthogonalization in the order first column, then second column, and so on, has the special block form below:

\begin{equation}\label {fast_decod_matrix}
\left(
\begin{array}{ccccccc}
 B_1 &   &  & &  & N_1\\
  &  B_2 &     & &  & N_2\\
  &   & \ddots    & & & N_3 \\
    & &   & & B_g & N_g\\
      & &   & &  & N_{g+1}\\
      & &   & &  & 
\end{array}
\right)
\end{equation} for some matrices $B_1$, $\dots$, $B_g$, and $N_1$, $\dots$, $N_{g+1}$.
Here, all empty spaces are filled by zeros, the $B_i$ are of size $n_i \times n_i$ and the $N_{i}$ are of size $n_{i} \times n_{g+1}$.

\end{theorem}

Before we prove this, we remark in more detail why previous authors have been interested in the special form of $R$ above: On applying the permutation $\pi$ to Equation \ref {basic_system}, we get
$\vecR(  Y ) = T_\pi\cdot \mbfs_\pi+ \vecR(N)$, and then, as in the beginning of this section, premultiplying by $Q^*$ we find $Q^*\vecR(  Y ) = R\cdot \mbfs_\pi+ Q^*\vecR(N)$.
It is clear from the block structure of the matrix $R$ that after fixing the values of the last $n_{g+1}$ variables in $\mbfs_\pi$, the remaining variables can be decoded in $g$ parallel steps, the $i$-th step involving $n_i$ variables.  The decoding complexity for this system is then of the order of $\Al^{n_{g+1}+\max{n_i}}$, where $\Al$ is the size of the effective real constellation $S$. This is in contrast to the complexity of $\Al^{2l}$ if the matrix $R$ has no special structure.

\begin{proof} If $X$ is fast decodable as per Definition \ref{deffd}, then as described in Remark \ref{rem:deffd}, the subsets $\Gamma_1, \ldots, \Gamma_g, \Gamma_{g+1}$  provide a permutation $\pi$ of $\{1, \dots, 2l\}$, and integers $g\ge 2$, $n_1$, $\dots, n_g, n_{g+1}$ with the properties described. 

Definition \ref{deffd} and Theorem \ref{orthog_means} also tell us that  every column of $T_\pi$ indexed by elements of 
$\pi^{-1}(\Gamma_i)$ is orthogonal to every column indexed by the elements of  $\pi^{-1}(\Gamma_j)$ ($1 \le i < j \le g$).
 It follows immediately that on applying a QR decomposition to $T_\pi$ in the order first column, then second column, etc., that the $R$ matrix, which results from the Gram-Schmidt orthogonalizations of the columns of $T_\pi$ in this order, will have the property that the columns indexed by $\pi^{-1}(\Gamma_i)$ will be perpendicular to those indexed by $\pi^{-1}(\Gamma_j)$.  This can be seen easily from how the Gram-Schmidt process works, but this can also be checked from the explicit form of the matrix $R$ obtained from this Gram-Schmidt orthogonalization, described for instance in \cite[Section III]{BHV} or \cite[Section VI] {SR}.

As for the other direction, assume that there is a {permutation} $\pi$ of $\{1, \dots, 2l\}$ and 
integers $g\ge 2$, $n_i \ge 1$ ($i=1, \dots, g$), and $n_{g+1} \ge 0$,
with  $n_1 + \cdots + n_{g+1} = 2l$,  such that for all $H$, $T_\pi(H) = QR$, where $Q$ is unitary and $R$ has the form as in Equation (\ref {fast_decod_matrix}) above.  Define the sets $\Gamma_i$ in terms of the integers $n_i$ as in Remark \ref{rem:deffd}, namely $\Gamma_1 = \pi(\{1, \ldots, n_1\})$ is the image of the first $n_1$ elements $\{1, \ldots, n_1\}$, $\Gamma_2$ is the image of the next $n_2$ elements, and so on.  It is clear from the block form of $R$ that for any {$u\in \Gamma_i$ and $v\in \Gamma_j$ ($1 \le i < j \le 2l$), the $\pi^{-1}(u)$-th and $\pi^{-1}(v)$-th columns of $R$} are orthogonal as vectors in $\R^{2n^2}$.  Since $Q$ is unitary, the same holds for the matrix $T_\pi(H)$. {Equivalently, the $u$-th and $v$-th columns of $T$ are orthogonal for all $H$. }
Thus, by Theorem \ref{orthog_means}, $A_{u}$ and $A_{v}$ are mutually orthogonal, so $X$ is fast decodable as per Definition \ref{deffd}.
\end{proof}

We summarize what we have shown in the next corollary: 

{
\begin{corollary} The following are equivalent for disjoint subsets  $\Gamma_i , \Gamma_j \subset \{1, \ldots , 2 l\}$: 
\begin{itemize}
\item for all $u\in \Gamma_i$ and $v\in \Gamma_j$ 
$$  A_uA_v^* + A_vA_u^* = 0. $$
\item for all $u\in \Gamma_i$ and $v\in \Gamma_j$, the $u$-th and $v$-th columns of $T=T(H)$ are orthogonal as real vectors for any $H$.
\item there exists a permutation $\pi$ on the index set $\{1, \ldots, 2l\}$ so that 
such that the matrix $R$ 
arising as in the statement of Theorem \ref{fd_means_nice_T} has a zero block in the entries 
$(\pi^{-1}(\Gamma_i), \pi^{-1}(\Gamma_j))$ and 
$(\pi^{-1}(\Gamma_j), \pi^{-1}(\Gamma_i))$. 
\end{itemize}
\label {thm:FDMOequiv}
\end{corollary}
}

\begin{corollary} Definition \ref{deffd} of fast decodability is equivalent to one given in \cite[Definition 4]{JR}. 
\end{corollary}

\begin{remark} \label{rem_fd_sublattice} In the notation of Definition \ref{deffd}, let $L$ be the lattice in $\R^{2n^2}$ generated by the columns of $T = T(H)$, and let $L_i$  ($i=1, \dots, g$) be the sublattices generated by the basis vectors of $L$ coming from the columns in $\Gamma_i$ (of the permuted matrix $T_\pi$). Fast-decodability can clearly be rephrased as the presence of sublattices $L_i$ ($g \ge 2$) generated by subsets of the  basis vectors that are orthogonal to one another in $\R^{2n^2}$. Indeed, previous work on fast decodability can be described in this language: seeking large numbers of sublattices generated by basis vectors that are orthogonal to one another.
\end{remark}

\begin{defi} \label{defn_g_grp_decod} We say that the fast decodable code $X = \displaystyle\sum_{i=1}^{2l} s_i A_i$ is $g$-group decodable if it is fast (lattice) decodable and if {$\Gamma_{g+1}$  in Definition \ref{deffd} is empty, so the matrix  $R$ of Theorem \ref{fd_means_nice_T} has a block-diagonal form. }
\end{defi}
\begin{remark} \label {rem_g_grp_orthog_sum} As in the proof of Theorem \ref{fd_means_nice_T}, the block-diagonal structure of $R$ of a $g$-decodable code translates (via pre-multiplication by $Q$) to the partitioning of the columns of $T$ into $g$ groups, the columns from any one group being orthogonal to the columns in any other group. Since $T$ is the transmitted lattice matrix, we see that $g$-group decodability of the code is equivalent to the decomposition of the transmitted lattice into an orthogonal sum of smaller dimensional lattices generated by the basis vectors, no matter what the channel matrix $H$.
\end{remark}

%
%
%
\section{Bounds on decoding complexity for full-rate codes} \label{secn_mo_shmo}

In this section, we will analyze the mutual orthogonality condition $A_iA_j^* + A_jA_i^* = 0$ of Theorem \ref{thm:FDMOequiv} and show that for full-rate codes, the best possible decoding complexity is not better than $\Al^{n^2+1}$ where $\Al$ is the size of the effective real constellation, and that $g$-group decoding is in fact not possible for full-rate codes. But first, we formalize the notion of decoding complexity:

\begin{defi} \label{def:decod_complexity}
The decoding complexity of the fast decodable space time code $X = \displaystyle\sum_{i=1}^{2l} s_i A_i$ is
defined to be $\Al^{n_{g+1} + \max_{1\le i \le g}{n_i}}$, where $n_i = |\Gamma_i|$, the $\Gamma_i$ as in Definition \ref{deffd}.
\end{defi}

Before delving into the main results of this section, we find it convenient to first gather a few lemmas concerning mutually orthogonal matrices that will be useful both here and in later sections.

\begin{lemma} \label{MO_under_matrix_mult} If matrices $A$ and $B$ are mutually orthogonal, so are $MA$ and $MB$ for any matrix $M$. If $M$ is invertible, then $A$ and $B$ are mutually orthogonal if and only if $MA$ and $MB$ are mutually orthogonal.
\end{lemma}

\begin{proof} This is a simple computation.
\end{proof}

\begin{lemma} \label{MO_yields_SH} If $A$ and $B$ are mutually orthogonal and $A$ is invertible, then $A^{-1}B$ is skew-Hermitian.
\end{lemma}

\begin{proof} By Lemma \ref{MO_under_matrix_mult} above, $A^{-1}A=I_n$ and $A^{-1}B$ are mutually orthogonal.  Writing down the mutual orthogonality condition for these two matrices, we find that $A^{-1}B$ is skew-Hermitian.
\end{proof}

\begin{lemma} \label{lem:mo_to_ac}
The $g$ invertible matrices $A_1=I_n, A_2,\ldots,A_g\in\mathcal{A} \subseteq M_n(\C)$ are mutually orthogonal if and only if $A_i$ is skew-Hermitian for $i\geq 2$ and $A_2$, $\dots$, $A_g$ pairwise anticommute.
\end{lemma}

\begin{proof} 
Assume that $A_1=I_n, A_2,\ldots,A_g\in\mathcal{A} \subseteq M_n(\C)$ are mutually orthogonal. Since $I_n$ and $A_i$ are mutually orthogonal for $i\geq 2$, we find that 
$A_i$ is skew-Hermitian for $i\geq 2$.
In particular, for $i,j\geq 2, i\neq j$, we may replace $A_i^*$ by $-A_i$ and $A_j^*$  by $-A_j$ in the orthogonality relation to obtain the anticommuting relation   $A_iA_j+A_jA_i=0$.
Conversely, assume that $A_i$ is skew-Hermitian for $i\geq 2$ and $A_2$, $\dots$, $A_g$ pairwise anticommute.  We  clearly have 
$I_n A_i^*+A_i I_n=0$ for $i \ge 2$.  Using the \SH \ relation to replace the second factor in each summand of $A_iA_j + A_jA_i$ by the negative of its conjugate transpose, we find that the $A_i$, for $i = 2, \dots, g$ are mutually orthogonal.
\end{proof}

Our first result is the following:
\begin{theorem} \label {thm:max_in_groups} 
Assume that the code $X = \displaystyle\sum_{i=1}^{2l} s_i A_i$ admits fast decodability, and  
let $k = \displaystyle\min_{1\le i \le g}{n_i}$, where $n_i = |\Gamma_i|$, the $\Gamma_i$  as in Definition \ref{deffd}.
Then $n_1 + \cdots + n_g \le n^2 + k$.
\end{theorem}
\begin{remark} \label {rem_less_if_k>2} In fact, we'll see later that if $k \ge 2$, then the sum $n_1 + \cdots + n_g \le n^2 + k-1$.
\end{remark}

We immediately get a high lower bound on the decoding complexity for full-rate codes from this theorem:
\begin{corollary} \label{rem_immediate_bound} 
The decoding complexity of a full-rate code of $n\times n$ matrices is at least $\Al^{n^2}$.  
 \begin{proof}
 Since a full-rate code has exactly $2n^2$ basis matrices, this theorem shows that the subset $\Gamma_{g+1}$ in Remark \ref{rem:deffd} must be 
of size at least $n^2 - k$, where $k = \displaystyle\min_{1\le i \le g}{n_i}$.  Having conditioned the symbols corresponding to $\Gamma_{g+1}$, decoding the first $g$ groups of symbols in parallel has a decoding complexity at least $\Al^k$, therefore the decoding complexity of the entire code must be at least $$\Al^{n^2-k}\cdot \Al^k = \Al^{n^2}.$$  \end{proof}
\end{corollary}
We will show later that the bound is actually higher: it is $\Al^{n^2+1}.$ 
\begin{corollary} \label{corol:num_grps_atmost_2} A full-rate code cannot be $g$-group decodable for $g \geq 3$.
\end{corollary}
\begin{proof}
For, if a code is $g$-group decodable, then, written in the notation of Theorem \ref{thm:max_in_groups}, we have $2n^2 = n_1 + \cdots +n_g \leq n^2+k$, by the theorem. So $n^2 \leq k$, the number of elements in the smallest block, implying there can be at most $2$ blocks. 
\end{proof}
We will see later that $2$-group decodability is also not possible for full-rate codes.

We now prove the theorem. 
\begin{proof}[Proof of Theorem \ref{thm:max_in_groups}] Let us denote the basis matrices in the groups $\Gamma_i$ ($i=1,\dots, g$) by $A_{i,j}$, $j = 1, \dots, n_i$. Multiplying the matrices on the left by any one $A_{i,j}^{-1}$ (recall from the beginning of Section \ref{secn_sys_model} that we assume that the basis matrices are invertible), we replace one of the matrices in the $i$-th block by the identity matrix $I_n$, and as for the modified matrices in the  other blocks, they are now orthogonal to $I_n$ by Lemma \ref{MO_under_matrix_mult} above. By Lemma \ref {MO_yields_SH} above, the  modified matrices  $A_{i,j}^{-1} A_{k,l}$ in the remaining blocks are all skew-Hermitian as well.  Since the remaining matrices $A_{i,j}^{-1} A_{k,l}$ are also $\R$-linearly independent by Lemma \ref{lem:mo_implies_li}, and since the dimension of the space of skew-Hermitian $n\times n$ matrices over $\R$ is $n^2$ (Section \ref {secn_prelim}), we find that for each $i$, $(n_1 + \cdots + n_g) - n_i \le n^2$. The result now follows immediately.

\end{proof}

Our next few results will help us sharpen the bounds on decoding complexity we obtain from Theorem \ref{thm:max_in_groups} (see Corollary \ref{rem_immediate_bound}).

\begin{theorem} \label{maxSHMO}
There can be at most $n^2-1$ $\R$-linearly independent matrices in $M_n(\C)$ that are both \SH  \ and \MO.
\end{theorem}
\begin{proof} For, suppose to the contrary that $A_1, \dots, A_{n^2}$ were $\R$-linearly independent, \SH, and \MO. The matrix $\imath I_n$ is \SH.  Suppose first that one of these $A_i$, say $A_1$, is an $\R$-multiple of $\imath I_n$. This is already a contradiction, since $A_1 A_2^*$ is \SH \ by the mutual orthogonality condition, but $A_1 A_2^*$ is a real multiple of  $\imath A_2^*$ and  is therefore Hermitian.  Now suppose that no $A_i$ is an $\R$-multiple of $\imath I_n$. The matrix $\imath I_n$, being \SH, can be written as a linear combination of these matrices $A_i$ since they form a basis for the \SH \ matrices, so $\imath I_n = \displaystyle\sum a_j A_j$ for real $a_j$. Now $A_1$ is not a real multiple of $\imath I_n$ by assumption. Consider $\imath I_n A_1^*$.  This is Hermitian.  On the other hand, $( \displaystyle\sum a_j A_j) A_1^* = a_1 A_1A_1^* + ( \displaystyle\sum a_jA_j)A_1^*$, where this second sum runs from $j = 2$ onwards.  But for $j = 2$ onwards, $A_j A_1^*$ is \SH \ by the mutual orthogonality condition, while both $\imath A_1^*$ and $a_1 A_1A_1^*$ are Hermitian.  For this to happen, $( \displaystyle\sum a_j A_j) A_1^* $, where the sum is over $j \ge 2$, must be zero, and $\imath A_1^* $ must equal $a_1 A_1A_1^*$. On canceling $A_1^*$ (recall our assumption that the basis matrices are invertible), we find that $A_1$ is a multiple of $\imath I_n$, contradiction.
\end{proof}

\begin{example} In the $2\times 2$ matrices $M_2(\C)$ over the complex numbers $\C$, consider the three matrices $A_1 = \left(
\begin{array}{cc}
 \imath & 0     \\
 0 & -\imath     
\end{array}
\right)$,
$A_2 = 
\left(
\begin{array}{cc}
 0 & -1     \\
 1 & 0     
\end{array}
\right)$, and 
$ A_3 = 
\left(
\begin{array}{cc}
 0 & -\imath     \\
 -\imath & 0     
\end{array}
\right)$. These three matrices are $\R$-linearly independent, skew-Hermitian, and pairwise mutually orthogonal matrices. Together with the identity matrix $I = 
\left(
\begin{array}{cc}
 1& 0     \\
 0 & 1     
\end{array}
\right)
$, 
they form a $\C$-basis for $M_2(\C)$, and as can be checked, no $\C$-linear combination of $I$, $A_1$, $A_2$, and $A_3$ is both skew-Hermitian and mutually orthogonal to $A_1$, $A_2$, and $A_3$.  Thus, the $2^2-1$ matrices $A_1$, $A_2$, and $A_3$ exemplify the contention of this theorem.
\end{example}

We get a quick corollary from this that we will sharpen considerably in the next section:
\begin{corollary}[See Corollary \ref{cor_max_groups_Greg} in Section \ref {AzAlg}] \label {max_g} For
{a code generated by invertible $n \times n$ matrices}, the maximum number of groups $g$ in notation of Definition \ref{deffd} is $n^2$.
\end{corollary}
\begin{proof} If the number of groups is more than $n^2$, then we can find $n^2+1$ matrices that are $\R$-linearly independent and \MO. Multiplying this set on the left by the inverse of one of them (as in the proof of Theorem \ref {thm:max_in_groups} above), we find $n^2$ \SH \ and \MO \  $\R$-linearly independent matrices, a contradiction.
\end{proof}

\begin{lemma} \label{what_if_rest_n^2}
{If any $g-1$ of the groups $\Gamma_1, \ldots, \Gamma_g$ from Definition \ref{deffd} together have at least $n^2$ matrices in them, then they have exactly $n^2$ elements in them, while the remaining group can only have one matrix in it.}
\end{lemma}

\begin{proof} Say the last $g-1$ groups, for simplicity, together have at least $n^2$ matrices, and 
suppose that the first group has at least two elements, call them $A$ and $B$.  By multiplying throughout by $A^{-1}$, we can assume that the two elements are $I$ and $B$.  Note that after multiplying by $A^{-1}$, because of the mutual orthogonality condition, the matrices in the remaining groups all become \SH \ (as in the proof of Theorem \ref {thm:max_in_groups} above).  Because there are at least $n^2$ \SH \ ($\R$-linearly independent) matrices, we find that there must be exactly $n^2$ of them because the dimension of the \SH \ matrices is $n^2$.  Call these $n^2$ matrices $C_1, \dots, C_{n^2}$.  We must have $\imath I_n$ in the linear span of these $C_i$ because $\imath I_n$ is also \SH. Thus, $\imath I_n = \displaystyle\sum a_i C_i$.  Now multiply on the right by $B^*$, where $B$ is as above.  Each of the products $C_i B^*$ is \SH \ because of the mutual orthogonality condition that requires  $C_iB^*+BC_i^*=0$.  Thus,  $\imath B^*$ is also \SH. It follows from this that $B^*$ is Hermitian, i.e., $B$ is Hermitian.  But now, we consider $C_i B^*$ for any $i$.   The mutual orthogonality condition says that this is \SH, so it equals $ -(B C_i^*)$, and since $C_i^*$ is \SH, this equals $BC_i$.  On the other hand, we just saw that $B$ is Hermitian, so $C_i B^* = C_i B$.  Thus, $B$ commutes with all $C_i$, i.e, with all  \SH \ matrices.  But this means $B$ commutes with all the Hermitian matrices as well, because every Hermitian matrix is of the form $\imath$ times a \SH \ matrix.  Thus, $B$ commutes with all matrices, and is Hermitian, so it must be a real scalar matrix.  But this violates the fact that $I_n$ and $B$ were two linearly independent matrices in the first group.
\end{proof}

\begin{corollary} \label{what_if_k_ge_2} 
If, as in the notation of Definition \ref{deffd}, $n_i \ge 2$ for any $i$, then the total number of matrices in the $g$ groups is at most $n^2+n_i-1$. In particular, if $k = \min n_i \ge 2$, then the total number is at most $n^2-1+k$.
d
\end{corollary}
\begin{proof} Since the $i$-th group has size $n_i \ge 2$, the remaining groups must have less than $n^2$ matrices in them, or else, the lemma above will be violated. It follows that there at most $n^2+n_i - 1$ matrices in the $g$ groups.
\end{proof}

We are now ready to sharpen the results we got in Corollary \ref{rem_immediate_bound}.

\begin{theorem} \label {thm:bad_fd_complexity} The decoding complexity of a full-rate space time code $X = \displaystyle\sum_{i=1}^{2n^2} s_i A_i$ is not better than $\Al^{n^2+1}$, where $\Al$ is the size of the effective real constellation.
\end{theorem}
\begin{proof}
Consider the basis matrices $A_i$: if there are at least two \MO \ groups, then, by Definition \ref{deffd}, the code is fast decodable, and by Theorem \ref{fd_means_nice_T} the $R$ matrix that comes from $T=T(H)$ will have the form (\ref{fast_decod_matrix}). Consider the integers $n_i$, notation as in Defintion \ref{deffd}. If any $n_i \ge 2$, then by Corollary \ref {what_if_k_ge_2}, the total number of matrices in the $g$ groups is at most $n^2 + n_i -1$.  Thus, the matrix $N_{g+1}$ in  (\ref{fast_decod_matrix}) will be of size at least $(n^2 -n_i+1) \times (n^2-n_i+1)$. Exactly as in the proof of Corollary \ref{rem_immediate_bound}, we find that the decoding complexity must be at least $\Al^{n^2-n_i+1}\cdot \Al^{n_i} =\Al^{n^2+1}$. If on the other hand all $n_i=1$, then we have $g$ groups of size $1$ each.  By  Corollary \ref{max_g}, $g \le n^2$, so $N_{g+1}$ is at least of size $n^2 \times n^2$. Thus, there are at least $n^2$ variables corresponding to $N_{g+1}$ that need to be conditioned, and then, the $g$ blocks are decoded in parallel, with complexity $\Al$ each.  Thus the decoding complexity is at least $\Al^{n^2}\cdot \Al = \Al^{n^2+1}$. 
\end{proof}
\begin{example}\label{ex:bound_is_strict} \textit{Silver Code:} This $2\times2$ code for four complex signal elements $s_1$, $s_2$, $s_3$, $s_4$ is given by $X(s_1, s_2) + TX(z_1,z_2)$, where for any $a$ and $b$, $X(a,b)=
\left(
\begin{array}{cc}
 a & -b^*     \\
 b & a^*     
\end{array}
\right)
$, and $T = \left(
\begin{array}{cc}
 1 & 0     \\
 0 & -1     
\end{array}
\right)$. The signal elements $s_3$ and $s_4$ are related to $z_1$ and $z_2$ by $(z_1, z_2)^T = M(s_3, s_4)^T$, where $M = \dfrac{1}{\sqrt{7}}\left(
\begin{array}{cc}
 1+\imath & -1 + 2\imath    \\
 1+2\imath & 1-\imath     
\end{array}
\right)$.
This code has a decoding complexity of at most $|S|^5$ (see \cite{JR} for instance). This example thus shows that our bound $n^2+1$ is strict. Moreover, Theorem \ref{thm:bad_fd_complexity}  shows that the Silver code cannot have a lower lattice decoding complexity than the known $|S|^5$.
\end{example}

\begin{theorem} \label {thm:no_g_grp_decod} It is not possible to arrange for the full-rate space-time code $X = \displaystyle\sum_{i=1}^{2n^2} s_i A_i$ to have  $g$-group decodability for any $g$.
\end{theorem}
\begin{proof} We have already seen in Corollary \ref{corol:num_grps_atmost_2} that $g$-group decodability is not possible for $g\ge 3$.  For $g=2$, note that one of two groups must have at least $n^2$ matrices. It follows from Lemma \ref{what_if_rest_n^2} that this group must have exactly $n^2$ elements and the other group must have only one element.  Since $n\ge 2$ in the space-time block code paradigm, $1+n^2 < 2n^2$, and $2$-group decodability is hence impossible.
\end{proof}
\begin{remark} \label {rem_indec_lattice} In a different language (see Remark \ref{rem_g_grp_orthog_sum}), Theorem \ref {thm:no_g_grp_decod} says that the transmitted lattice of a full-rate space-time code does not split off as an orthogonal sum of smaller dimensional lattices generated by the canonical basis vectors.
\end{remark}
%
%
\section{Azumaya Algebras and Bounds on the Number of Groups}\label{AzAlg}

In this section, we will delve into the arithmetic of central-simple algebras, using machinery from commutative ring theory and Azumaya algebras, to determine significantly small upper bounds on the number of orthogonal sublattices generated by the basis vectors of the transmitted lattice $T = T(H)$, or what is the same, the number of blocks $g$ of the $R$ matrix in Equation (\ref{fast_decod_matrix}).  We had already derived an upper bound of $n^2$ for full-rate codes in Corollary \ref {max_g}, but as we will see, this bound is too high. In fact, the bound behaves more like $\log_2(n)$ (see Theorem \ref{Greg_hauptsatz} for a precise statement).   The bound we derive in this section will be independent of the code rate ($l$).  Since the matrices in  distinct groups are pairwise mutually orthogonal, we will derive our bound by answering the following question: How many $\R$-linearly independent pairwise mutually orthogonal matrices can we find in $M_n(\C)$?  In fact, we will actually answer a broader question: Let $k\subset {\mathbb{C}}$ be a number field, let $\mathcal{A}$ be a central simple $k$-subalgebra of $M_n({\mathbb{C}})$. How many $\R$-linearly independent pairwise mutually orthogonal matrices can we find in the subalgebra $\mathcal{A} \subseteq M_n(\C)$?
(Of course, by Lemma \ref{lem:mo_implies_li}, we may drop the requirement that the matrices be $\R$-linearly independent.)

As in the earlier sections, we will assume that our pairwise orthogonal matrices are all invertible. Note that if a matrix $A\in \mathcal{A} \subseteq M_n(\C)$ is invertible as an element of $M_n(\C)$, its inverse must actually lie in $\mathcal{A}$. This is because $A^{-1}$ can be obtained from the minimal polynomial of $A$ over $k$ as follows: if the minimal polynomial is $A^t + k_{t-1} A^{t-1} + \cdots k_1A + k_0$, then $k_0 \neq 0$ because $A$ is invertible as a matrix, so the inverse of $A$ can be written by factoring out $A$ as $(-1/k_0)( A^{t-1} + k_{t-1}A^{t-2} + \cdots + k_1)$. The inverse of $A$  hence lives in the subalgebra $k[A] \subseteq \mathcal{A}$.)

All the $k$-algebras we consider will be implicitly assumed to be finite-dimensional over $k$. Various background facts about commutative rings and Azumaya algebras are collected in Appendices \ref{app_comm_alg} and \ref{app_azu} respectively. We will assume basic knowledge of central simple algebras (see \cite {BOText} for instance).

Lemmas \ref {MO_under_matrix_mult}, \ref {MO_yields_SH}, and \ref{lem:mo_to_ac} show us that the existence of  (invertible) \MO \ matrices $A_i$, $i=1, \dots, m$  is equivalent (upon replacing the $A_i$ by say $A_1^{-1}A_i$) to the existence of matrices $C_i = A_1^{-1}A_i$, $i=2, \dots, m$ which are skew-Hermitian and anticommute pairwise.

So, focusing on the necessary anticommuting condition  above, we study the following question. (In the sequel, $\Ac^\times$ will refer to the invertible elements of $\Ac$.)

{\bf Question. }Let $k$ be a number field, and let $\Ac$ be a central simple $k$-algebra. How many elements $u_1,\ldots,u_{r}\in \Ac^\times$ which pairwise anticommute can we find?

We now investigate this question.

Once and for all, we fix a central simple $k$-algebra $\Ac$, and we assume to have elements $u_1,\ldots,u_r\in \Ac^\times$ such that $u_iu_j+u_ju_i=0 \ \mbox{ for all }i\neq j$,
for some $r\geq 2.$ For the moment, we only assume that $k$ is any field of characteristic different from $2$.

Notice that $u_i$ and $u_j^2$ commute for all $i,j$. Indeed, this is clear if $i=j$, and if $i\neq j$, we have 
$u_iu_j^2=-u_j u_iu_j=u_j^2 u_i$.

This implies that $u_i^2,u_j^2$ commute for all $i,j$. 
Consequently, the $k$-algebra $$R=k[u_1^2,u_1^{-2},\ldots,u_r^2,u_r^{-2}]$$ 
is a commutative $k$-subalgebra of $\Ac$. (Of course, as remarked in the second paragraph of this section, the $k$-algebra generated by $u_i^2$ will already contain $u_i^{-2}$, but we choose to include the $u_i^{-2}$ in the generators of $R$ to emphasize that the $u_i^2$ are units in $R$, a fact we will need below.)

Notice also that for any $u_{i_1},\ldots,u_{i_k}$, we have $(u_{i_1}\cdots u_{i_k})^2=\pm u_{i_1}^2\cdots u_{i_k}^2\in R^\times$.

We recall the definition of the algebra $(a,b)_R$ from Part \ref{ex_quat_over_R}  of Examples \ref{ex_Az_alg}
 in Appendix \ref {app_azu}:  given a commutative ring $R$ and $a$, $b$ in $R^\times$, $(a,b)_R$ is the $R$-algebra generated by two elements $e$ and $f$ subject to the relations $e^2=a$, $f^2=b$, and $fe = -ef$. It has the matrix realization described in Appendix \ref {app_azu}.

\begin{lemma} \label {lemma_lotsa_quats}
Let $r=2s$ or $2s+1$. Keeping notation above, $\Ac$ contains a subring isomorphic to  $$(a_1,b_1)_R\otimes_R\cdots\otimes_R (a_s,b_s)_R,$$
for some $a_p,b_p\in R^\times.$
\end{lemma}  
  
\begin{proof}
If $I$ is any subset of $\{1,\ldots, n\}$, set $u_I=\ds \prod_{i\in I}u_i$. It is then easy to check that for all $I,J$, we have $u_Iu_J=(-1)^{\vert I\vert\cdot \vert J\vert -\vert I\cap J\vert}u_J u_I$.

For $p=1,\ldots,s$, 
set $$I_p=\{1,\ldots,2p-1\}, \quad J_p=\{1,\ldots,2p-2,2p\}.$$
We then have $\vert I_p\vert =\vert J_p\vert =2p-1, \vert I_p\cap J_p\vert =2p-2$, and for all $1\leq p< q\leq s$, we have 
$\vert I_p\cap I_q\vert=\vert J_p\cap J_q\vert=\vert I_p\cap J_q\vert =\vert I_q\cap J_p\vert=2p-1$.
 
Now set $$\alpha_p=u_{I_p}, \quad \beta_p=u_{J_p}.$$

Notice that $a_p=\alpha_p^2,b_p=\beta_p^2\in R^\times$. 
Moreover, for all $p=1,\ldots,s$, we have 
$\alpha_p\beta_p=$ $u_{I_p}u_{J_p}=(-1)^{(2p-1)^2-(2p-2)}u_{J_p}u_{I_p}$  $=-u_{J_p}u_{I_p}$  $=-\beta_p\alpha_p$.
Thus, for all $p=1,\ldots,s$, we have an $R$-algebra morphism $\varphi_p:(a_p,b_p)_R\to \Ac$, which maps the generators $e_p$ and $f_p$ onto $\alpha_p$ and $\beta_p$ respectively.

Now for all $1\leq p<q\leq s$, we have 
$\alpha_p\alpha_q=$ $u_{I_p}u_{I_q}=$ $(-1)^{(2p-1)(2q-1)-(2p-1)}u_{J_p}u_{I_p}$ $=\alpha_q\alpha_p$.
Similarly, we have $\beta_p\beta_q=\beta_q\beta_p$.
We also have $\alpha_p\beta_q=$ $u_{I_p}u_{J_q}=$ $(-1)^{(2p-1)(2q-1)-(2p-1)}u_{J_q}u_{I_p}=$ $\beta_q\alpha_p$.
Similarly, we have $\alpha_q\beta_p=\beta_p\alpha_q$.

It follows that $\varphi_1,\ldots,\varphi_s$ have pairwise commuting images. Thus, they induce an $R$-algebra morphism $$(a_1,b_1)_R\otimes_R\cdots\otimes_R (a_s,b_s)_R\to \Ac.$$
By Lemma \ref{inj} and Remark \ref{reminj}, this morphism is injective.
\end{proof}

We may now give a full answer to the previous question. 
\begin{theorem} \label {thm: max_Greg_anticommute}
Let $k$ be a number field, and let $\Ac$ be a central simple $k$-algebra. Let $u_1,\ldots,u_r$ ($r\ge 2$) be invertible elements in  $\Ac$ which pairwise anticommute. Then we have $$r\leq 2\nu_2\left(\frac{\deg(\Ac)}{{\rm ind}(\Ac)}\right)+2 \mbox{ if $r$ is even}$$ and $$r\leq 2 \nu_2\left(\frac{\deg(\Ac)}{{\rm ind}(\Ac)}\right)+3 \mbox{ if $r$ is odd},$$ where $\nu_2$ denotes the $2$-adic value of an integer, i.e., the highest power of $2$ that divides that integer.

In particular, if $\Ac$ is a central division $k$-algebra, then $r=2,3.$
\end{theorem}

\begin{remark} See Appendix \ref{app_RHE_connections} for how this result above compares with the classical Hurwitz-Radon-Eckmann bound on anticommuting matrices.
\end{remark}

\begin{proof}
We may assume $r>2$ if $r$ is even, and $r > 3$ if $r$ is odd, since otherwise this is trivial.
Write $r=2s$ or $r=2s+1$, so $s \ge 2$.
By the previous lemma, $\Ac$ contains an $R$-algebra isomorphic to $$(a_1,b_1)_R\otimes_R\cdots\otimes_R (a_s,b_s)_R,$$
for some $a_p,b_p\in R^\times.$ By Proposition \ref{prodquat}  in Appendix \ref {app_azu} applied $s-1$ times (note that $s-1 \ge 1$ by assumption), this $R$-algebra is isomorphic to $$M_{2^{s-1}}(R)\otimes_R (c,d)_R\cong_R (M_{2^{s-1}}(k)\otimes_k R)\otimes_R (c,d)_R$$ for some $c,d\in R^\times.$ 
Hence $\Ac$ contains a $k$-subalgebra isomorphic to $M_{2^{s-1}}(k)$. The centralizer theorem then implies that $$\Ac\cong_k M_{2^{s-1}}(k)\otimes_k \Ac',$$ for some central simple $k$-algebra $\Ac'$, which is Brauer-equivalent to $\Ac$ by definition. Therefore, we may write $$\Ac\cong_k M_\ell(D), \quad \Ac'\cong_k M_t(D),$$ where $D$ is a central division $k$-algebra.
Thus, we get $$M_\ell(D)\cong_k M_{2^{s-1}t}(D),$$
and then $2^{s-1}t=\ell=\frac{\deg(\Ac)}{{\rm ind}(\Ac)}.$
The desired result follows easily.
\end{proof}

\begin{remark}
If $\Ac$ is a central simple $k$-algebra of odd degree, then $\Ac$ does not contain pairwise anticommuting invertible elements. 

Indeed, if $u_1$ and $u_2$ anticommute, then we have $$\Nrd_\Ac(u_1u_2)=\Nrd_\Ac(u_1)\Nrd_\Ac(u_2)=\Nrd_\Ac(-u_2u_1)=-\Nrd_\Ac(u_2)\Nrd_\Ac(u_1),$$
where the last equality arises from the fact that $\Nrd_\Ac(-1) = -1$ since $\Ac$ has odd degree.
Hence, $\Nrd_\Ac(u_1)\Nrd_\Ac(u_2)=0$. But the reduced norm of an invertible element of $\Ac$ is non-zero, hence a contradiction.
\end{remark}

Hence the previous bounds are not always sharp.
However they may be sharp in certain cases as the following example shows, which proves that these bounds are the best possible ones.

\begin{example}\label{exsharp}
Let $\ell\geq 0$ be an integer, let $Q=(a,b)_k$ be a division quaternion $k$-algebra, and let $\Ac=M_{2^\ell}(k)\otimes_k Q$.

In order to avoid mixing notation, we will denote exceptionally by $\odot$ the Kronecker product of matrices. If $t\geq 0$ is an integer, we denote by $M^{\odot t}$ the Kronecker product of $t$ copies of $M$, where $M^{\odot 0}$ is the identity matrix by convention.

Let $$H_1=\begin{pmatrix}1 & 0 \cr  0 & -1\end{pmatrix} \ \mbox{ and } \ H_{-1}=\begin{pmatrix}0 & -1 \cr  1 & 0\end{pmatrix}.$$
For $p=1,\ldots,\ell,$ set $$U_{2p-1}=H_1^{\odot (p-1)}\odot H_1H_{-1}\odot I_2^{\odot (\ell-p)} \ \mbox{and } \  
U_{2p}=H_1^{\odot (p-1)}\odot H_{-1}\odot I_2^{\odot (\ell-p)}.$$

The properties of the Kronecker product and the fact that $H_1H_{-1}=-H_{-1}H_1$, show that $U_1,\ldots,U_{2p}$ are invertible matrices of $M_{2^\ell}(k)$ which pairwise anticommute.

Now let $e$ and $f$ be the generators of $Q$. Then it is easy to check that the $2\ell+3$ invertible elements 

$$U_1\otimes 1,\ldots,U_{2\ell}\otimes 1, U_{1}\cdots U_{2\ell}\otimes e,U_{1}\cdots U_{2\ell}\otimes f, U_{1}\cdots U_{2\ell}\otimes ef\in \Ac$$ pairwise anticommute.   

Notice for later use that $U_{2p-1}$ is symmetric and $U_{2p}$ is skew-symmetric for $p=1,\ldots,\ell$. Notice also that $U_1\cdots U_{2\ell}$ is symmetric, as a straightforward computation shows. 
\end{example}

As a corollary, we also get an answer to our main problem.

\begin{corollary} \label {cor_max_groups_Greg}
Let $k$ be a number field, let $\mathcal{A}$ be a central simple $k$-subalgebra of $M_n({\mathbb{C}})$.
Assume that we have $g$ non-zero matrices $A_1,\ldots,A_g\in\mathcal{A}^\times$ $(g\geq 2)$ such that $$A_i^*A_j+A_j^*A_i=0 \ \mbox{for all }i\neq j.$$
Then $g\leq  2\nu_2(\frac{\deg(\Ac)}{{\rm ind}(\Ac)})+3$ if $g$ is odd, and $g\leq  2\nu_2(\frac{\deg(\Ac)}{{\rm ind}(\Ac)})+4$ if $g$ is even. 

In particular, if $\Ac$ is a central division $k$-algebra, then $g \leq 4.$ 
\end{corollary}

\begin{proof}
By Lemma \ref{lem:mo_to_ac}, the existence of $g$ such matrices implies the existence of $g-1$ invertible elements of $\mathcal{A}$ which pairwise anticommute. Now apply the previous theorem to conclude.
\end{proof}

The next example shows that these bounds may be sharp.

\begin{example}
Let $k\subset \rr$, and let $U_1,\ldots, U_{2\ell}\in M_{2^\ell}(k)\subset M_{2^\ell}(\rr) $ be the matrices introduced in Example \ref{exsharp}. Set $Q=(-1,-1)_k$, so that $Q$ is a division $k$-algebra.

The multiplication matrices of $e$ and $f$ with respect to the $k(i)$-basis $(1,e)$ of $Q$ (viewed as a right $k(i)$-vector space) are the skew-Hermitian matrix $iH_1$ and the hermitian matrix $H_{-1}$ respectively. Notice that $iH_1H_{-1}$ is skew-Hermitian. The results of Example \ref{exsharp} show that the matrices $$U_1\odot I_2,\ldots, U_{2\ell}\odot I_2, U_1\cdots U_{2\ell}\odot (iH_1),U_1\cdots U_{2\ell}\odot H_{-1}, U_1\cdots U_{2\ell}\odot (iH_1H_{-1})$$ pairwise anticommute.

Each of these matrices are hermitian or skew-Hermitian. Multiplying by $i$ the appropriate matrices yields a set of $2\ell+3$ skew-Hermitian matrices which pairwise anticommute. More precisely, one may check that the matrices 
$$U_{2p-1}\odot I_2, U_{2p}\odot (iI_2), p=1,\ldots,\ell,$$
$$U_1\cdots U_{2\ell}\odot (iH_1),U_1\cdots U_{2\ell}\odot (iH_{-1}), U_1\cdots U_{2\ell}\odot (iH_1H_{-1})$$
are skew-Hermitian matrices which pairwise anticommute. Adding the identity matrix then gives rise to a set of $2\ell+4$ mutually orthogonal matrices.
\end{example}

It is worth rewording the result in Corollary \ref{cor_max_groups_Greg} in the language of our space-time code. We have the following:
\begin{theorem} \label {Greg_hauptsatz} 
If the space-time code $X = \displaystyle\sum_{i=1}^{2l} s_i A_i$  is fast-decodable, then the number of groups $g$ in (\ref{fast_decod_matrix}) is at most $2 \nu_2(n) + 4$.  If we assume that the $A_i$ are chosen from some $k$-central simple algebra $\Ac \subseteq M_n(\C)$, where $k$ is some number field, then, this upper bound drops to $g \leq 2\nu_2(\frac{\deg(\Ac)}{{\rm ind}(\Ac)})+4$. In particular, if the $A_i$ are chosen from a $k$-central division algebra, then $g \le 4$.
\end{theorem}

We get an immediate corollary:
\begin{corollary} \label {cor_cheap_result}
The decoding complexity of a fast decodable space-time code $X = \displaystyle\sum_{i=1}^{2l} s_i A_i$ where the $A_i$ are chosen from a division algebra is at least $\Al^{\lceil l/2 \rceil}$.
\end{corollary}
\begin{proof} At least one of the groups $\Gamma_i$ ($i=1, \dots, g$) in Definition \ref{deffd} must be of size at least $\lceil 2l/4 \rceil$, as $g \le 4$ when the $A_i$ are chosen from a division algebra. Thus, the decoding complexity is at least $\Al^{n_{g+1} + \lceil 2l/4 \rceil} \ge \Al^{\lceil l/2 \rceil}$.
\end{proof}

\appendices
\section{Commutative Algebra} \label {app_comm_alg} 

We collect here some useful results in commutative algebra. We start with the notion of an Artin ring.

\begin{defi}
A commutative ring $R$ is an {\it Artin} ring if every descending chain of ideals $I_0 \supset I_1 \supset I_2 \supset \cdots$ of $R$ is stationary, i.e., there exists $n > 0$ such that $I_n = I_{n+1} = I_{n+2} = \cdots$.
\end{defi}

\begin{example}\label{exartin}
If $k$ is a field, any finite-dimensional commutative $k$-algebra $R$ is an Artin ring.
Indeed, any ideal is in particular a finite-dimensional $k$-subspace of $R$, so it cannot exist a strictly decreasing chain of ideals. 
\end{example}

\begin{theorem}\cite[Ch.8, Thm 8.5]{At}
Any Artin ring is Noetherian, that is every ideal is finitely generated.
\end{theorem}

\begin{corollary}\label{artinnil}
Let $R$ be a local Artin ring, with maximal ideal $\mm$. Then there exists $n\geq 1$ such that $\mm^n=0$.
\end{corollary}

\begin{proof}
By assumption, the descending chain of ideals $\mm\supset\mm^2\supset\cdots\supset \mm^n\supset\cdots$ is stationary, hence there exists $n\geq 1$ such that $\mm^{n+1}=\mm \cdot \mm^{n}=\mm^n$. Since $R$ is Noetherian by the previous theorem, $\mm$ is finitely generated, and since $R$ is local with unique maximal ideal $\mm$,    $\mm^n=0$  by Nakayama's lemma.
\end{proof}

We also have the following result.

\begin{theorem}\cite[Ch.8, Thm. 8.7]{At}\label{struc}
Any Artin ring is isomorphic to the direct product of finitely many Artin local rings. In particular, an Artin ring has finitely many maximal ideals.
\end{theorem}

We now define Hensel rings.

\begin{defi}
A commutative ring $R$ is a {\it Hensel} ring if $R$ is local, with maximal ideal $\mm$, and for any monic polynomial $f\in R[X]$ such that $\ov{f}=\ov{g}_0 \ov{h}_0\in R/\mm[X]$ for some coprime monic polynomials $\ov{g}_0,\ov{h}_0\in R/\mm[X]$, there exists coprime monic polynomials $g,h\in R[X]$ such that $f=gh$ and $\ov{g}=\ov{g}_0,\ov{h}=\ov{h}_0$.
\end{defi}

The following result is well-known.

\begin{proposition}\label{artinhensel}
Any local Artin ring is a Hensel ring.
\end{proposition}

\begin{proof}
Since the maximal ideal $\mm$ of a local ring is nilpotent by Corollary \ref{artinnil}, $R$ is canonically isomorphic to its $\mm$-completion, that is $R$ is complete. Since complete rings are Hensel rings by \cite[Prop. 4.5]{Mi}, we are done.
\end{proof}


\section{Azumaya Algebras} \label {app_azu}

We collect here some notions on Azumaya algebras that are needed in the paper. The word `algebra' implicitly means `associative algebra with unit'.

In this section, $R$ is a commutative ring with unit. We first define Azumaya $R$-algebras. The reader willing to learn more about Azumaya algebras will refer to \cite[III.5]{Kn}.

\begin{defi}
An {\it Azumaya} $R$-algebra is an $R$-algebra $A$, which is finitely generated as an $R$-module and such that $A\otimes_R R/\mathfrak{m}$ is a central simple $R/\mathfrak{m}$-algebra for every maximal ideal $\mathfrak{m}$ of $R$. 
\end{defi}

\begin{example}~ \label {ex_Az_alg}
\begin{enumerate}
\item Let $B$ be a central simple $k$-algebra, and let $R$ be a commutative $k$-algebra. Then $A=B\otimes_k R$ is an Azumaya $R$-algebra.
 
Indeed, since $B$ is finite dimensional over $k$, $B\otimes_k R$ is finitely generated as an $R$-module. Let $\mathfrak{m}$ be any maximal ideal of $R$. Since $R$ is a $k$-algebra, $k$ identifies to a subring of $R$, and we have a ring morphism $k\to R/\mathfrak{m}$ which is injective, since $k$ is a field. Hence $R/\mathfrak{m}$ is a field extension of $k$. 
Now we have $$A\otimes_R R/\mathfrak{m}=(B\otimes_k R)\otimes_R  
R/\mathfrak{m}\cong_{R/\mm} B\otimes_k R/\mathfrak{m}.$$
Since $B$ is a central simple $k$-algebra, $B\otimes_k R/\mathfrak{m}$ is a central simple $R/\mathfrak{m}$-algebra (see \cite[Corollary III.1.5 (2)]{BOText}) and we are done.

\item  If $A$ and $A'$ are Azumaya $R$-algebras, then $A\otimes_R A'$ is an Azumaya $R$-algebra. First, since $A$ and $A'$ are finitely generated as $R$-modules, so is $A\otimes_R A'$. Now for every maximal ideal $\mathfrak{m}$ of $R$,
we have $$(A\otimes_R A')\otimes_R R/\mathfrak{m}\cong_{R/\mm} (A\otimes_R R/\mathfrak{m})\otimes_{R/\mm} (A'\otimes_R R/\mathfrak{m}).$$
This last $R/\mathfrak{m}$-algebra is the product  of two central simple $R/\mathfrak{m}$-algebras be assumption, hence a central simple $R/\mathfrak{m}$-algebra by \cite[Corollary III.1.5 (1)]{BOText}.
 
\item  For all $n\geq 1$, $M_n(R)$ is an Azumaya $R$-algebra. Indeed, $M_n(R)$ is a finitely generated $R$-module, and for every maximal ideal $\mathfrak{m}$ of $R$, we have $$M_n(R)\otimes_R R/\mathfrak{m}\cong_{R/\mathfrak{m}}M_n(R/\mathfrak{m}),$$
which is central simple over $R/\mathfrak{m}$.

\item  \label {ex_quat_over_R} We will assume in this example that $R$ is such that for all maximal ideals $\mathfrak{m}$, $R/\mathfrak{m}$ is of characteristic not $2$. Let $a,b\in R^\times,$ and consider the $R$-submodule $(a,b)_R$ of $M_4(R)$ generated by the matrices 
$$I_4=\left(\begin{array}{cccc}1 & 0 & 0 & 0\cr 0 & 1 & 0 & 0 \cr 0& 0 & 1 & 0 \cr 0& 0&0&1 \end{array}\right) , e=\left(\begin{array}{cccc}0 & a & 0 & 0\cr 1 & 0 & 0 & 0 \cr 0& 0 & 0 & a \cr 0& 0&1&0 \end{array}\right) ,$$
$$ f=\left(\begin{array}{cccc}0 & 0 & b & 0\cr 0 & 0 & 0 & -b \cr 1& 0 & 0 & 0 \cr 0& -1&0&0 \end{array}\right) ,
ef=\left(\begin{array}{cccc}0 & 0 & 0 & -ab\cr 0 & 0 & b & 0 \cr 0& -a & 0 & 0 \cr 1& 0&0&0 \end{array}\right).$$ 
 
Straightforward computations show that these matrices are linearly independent over $R$, and that we have $$e^2=a,f^2=b,fe=-ef.$$
It easily follows that $(a,b)_R$ is a free $R$-module of rank $4$, which is an $R$-subalgebra of $M_4(R)$.
This $R$-algebra is denoted by $(a,b)_R$.

It can be viewed also as the $R$-algebra generated by two elements $e,f$ subject to the relations $$e^2=a,f^2=b, ef=-fe.$$

Then $(a,b)_R$ is an Azumaya $R$-algebra. Indeed, let $\mathfrak{m}$ be a maximal ideal of $R$. Since $a,b\in R^\times,$ $a$ and $b$ are non-zero elements of $R/\mathfrak{m}$. The explicit realization above shows easily that we have $$(a,b)_R \otimes_R R/\mathfrak{m}\cong_{R/\mathfrak{m}} (\ov{a}, \ov{b})_{R/\mathfrak{m}},$$ and it is well known that over a field of characteristic not $2$, the quaternion algebra generated by symbols $e$ and $f$ subject to $e^2=\ov{a},f^2=\ov{b}, ef=-fe$ is a central simple algebra. 
Hence the conclusion.
\end{enumerate} 
\end{example}

Azumaya algebras share common properties with central simple algebras. For example, we have the following lemma.

\begin{lemma}\label{inj}
Let $A$ and $B$ be two $R$-algebras. Assume that $A$ is an Azumaya $R$-algebra, and that $B$ is a faithful $R$-algebra, that is the $R$-algebra map $$\func{R}{B}{r}{r\cdot 1_B}$$ is injective. Then every $R$-algebra morphism $f:A\to B$ is injective.
\end{lemma}

\begin{proof}
Let $A,B$ and $f:A\to B$ as in the statement of the lemma. Then $\ker(f)$ is a two-sided ideal, hence an $A$-$A$-bimodule. By \cite[Ch. III, Theorem 5.1.1. (2)]{Kn}, $A$ is central, that is the $R$-algebra map $$\func{R}{Z(A)}{r}{r\cdot 1_A}$$ is an isomorphism, and separable, meaning that $A$ is a projective module for the natural $A\otimes_R A^{op}$-module structure induced by the multiplication map. By \cite[Corollary 3.7]{Dem}, there exists an ideal $I$ of $R$ such that $\ker(f)=I\cdot A$. 
Since $\ker(f)=I\cdot A$, for all $x\in I$, we have $$0_B=f(x\cdot 1_A )=x\cdot f(1_A)=x\cdot 1_B.$$
By assumption on $B$, we get $x=0$. Thus $I=0$, and $\ker(f)=0$.
\end{proof}

\begin{remark}\label{reminj}
If $B$ is any ring, and $R$ is a commutative subring of $B$, then the product law endows $B$ with the structure of an $R$-algebra satisfying the condition of the previous lemma, since for any $r\in R$, we have $r\cdot 1_B=r1_B=r$.
\end{remark}

The following result was proven in \cite[Theorem 32]{Az}, and will be useful to prove the next proposition.

\begin{theorem}
Let $R$ be a Hensel ring, with unique maximal ideal $\mathfrak{m}$. For every central simple $R/\mathfrak{m}$-algebra $B$, there exists an Azumaya $R$-algebra $A$, unique up to $R$-isomorphism, such that $A\otimes_R R/\mathfrak{m}\cong_{R/\mm} B.$
\end{theorem}

\begin{proposition}\label{artin}
Let $R$ be an Artin ring, and $A$, $B$ be Azumaya $R$-algebras. Then $A\cong_R B$ if and only if $A\otimes_R R/\mm\cong_{R/\mm} B\otimes_R R/\mm$ for every maximal ideal $\mm$ of $R$.
\end{proposition}

\begin{proof}
One implication is trivial. To prove the other one, notice that by Theorem \ref{struc}, we have a ring isomorphism $$\varphi: R\overset{\sim}{\to} R_1\times \cdots \times R_s,$$ for some local Artin rings $R_1,\ldots,R_s.$ 
We then have a 1-1-correspondence between the set of Azumaya $R$-algebras $A$ and the set of tuples $(A_1,\ldots,A_s)$, where $A_i$ is an Azumaya $R_i$-algebra, which is given by 
$$\begin{array}{ccc}A&\longmapsto& (A\otimes_R R_1,\ldots, A\otimes_R R_s)\cr (A_1\times \cdots \times A_s)\otimes_{R_1\times\cdots\times R_s} R&\longmapsfrom & (A_1,\ldots,A_s).\end{array}$$

Moreover, $A\cong_R B$ if and only if
$A\otimes_R R_i\cong_{R_i} B\otimes R_i$ for $i=1,\ldots,s.$

Let $\mm'_i$ be the maximal ideal of $R_i$. Then the ideal 
$$\mm_i=\varphi^{-1}(R_1\times\cdots\times R_{i-1}\times\mm'_i\times R_{i+1}\times \cdots\times  R_s)$$ is a maximal ideal of $R$, and the canonical projection $R\to R_i$ induces a ring isomorphism $$R/\mm_i\overset{\sim}{\to} R_i/\mm'_i.$$
This yields $$A\otimes_R R/\mm\cong_{R_i/\mm'_i}(A\otimes_R R_i)\otimes_{R_i}R_i/\mm'_i.$$
Hence, by assumption we get $$(A\otimes_R R_i)\otimes_{R_i}R_i/\mm'_i\cong_{R_i/\mm'_i} (B\otimes_R R_i)\otimes_{R_i}R_i/\mm'_i.$$

Since $R_i$ is a local Artin ring, it is a Hensel ring by Proposition \ref{artinhensel}. The previous theorem then shows that $A\otimes_R R_i\cong_{R_i} B\otimes R_i$. Since this is true for all $i=1,\ldots,s$, we get $A\cong_R B$ as required.
\end{proof}

As a consequence, we get the following proposition, which will be crucial for our coding considerations.

\begin{proposition}\label{prodquat}
Let $k$ be a number field, and let $R$ be a finite-dimensional commutative $k$-algebra. For all $a,b,a',b'\in R^\times$, there exist $c,d\in R^\times$ such that $$(a,b)_R\otimes_R (a',b')_R\cong_R M_2(R)\otimes_R (c,d)_R.$$
\end{proposition}

\begin{proof}
Notice first that $R$ is an Artin ring by Example \ref{exartin}. Let $\mm$ be a maximal ideal of $R$. Notice that  $R/\mm$ is an extension of $k$ of finite degree, the $k$-vector space structure being given by the map $k\to R\to R/\mm.$ Hence $R/\mm$ is a number field (and $(a,b)_R$, etc., are Azumaya algebras over $R$). Since the exponent and index of central simple algebras over a number field must be equal, and since the exponent of the tensor product of two quaternion algebras over  $R/\mm$ is at most $2$, the tensor product is of the form $M_2(B)$, where $B$ is either a division algebra of index $2$, and hence expressible as a quaternion algebra, or else, $B$ is itself $M_2( R/\mm)$, which is expressible as the quaternion $(1,1)_ {R/\mm}$. In either case, 
therefore, there exists $\ov{c}_\mm,\ov{d}_\mm\in (R/\mm)^\times$ such that $$\begin{array}{lll}((a,b)_R\otimes_R(a',b')_R)\otimes_R R/\mm & \cong_{R/\mm}& (\ov{a},\ov{b})_{R/\mm}\otimes_{R/\mm}(\ov{a'},\ov{b'})_{R/\mm} \cr 
& \cong_{R/\mm}& M_2(R/\mm)\otimes_{R/\mm}(\ov{c}_\mm,\ov{d}_\mm)_{R/\mm}\end{array}.$$

Since $R$ has finitely many maximal ideals by Theorem \ref{struc}, the Chinese Remainder Theorem
shows that there exist $c,d\in R$ such that $$c\equiv c_\mm \mod \mm \ \mbox{ and } \ d\equiv d_\mm \mod \mm$$ for all maximal ideals $\mm$ of $R$. Notice that $c,d\in R^\times$, since they do not belong to any maximal ideal of $R$ by construction.

For any maximal ideal $\mm$ of $R$, we then get $$\begin{array}{lll}((a,b)_R\otimes_R(a',b')_R)\otimes_R R/\mm & \cong_{R/\mm}&  M_2(R/\mm)\otimes_{R/\mm}(\ov{c},\ov{d})_{R/\mm} \cr & \cong_{R/\mm} & (M_2(R)\otimes_R (c,d)_R)\otimes_R R/\mm\end{array}.$$
Now apply the previous proposition to conclude.
\end{proof}

\section {Connections between Theorem \ref{thm: max_Greg_anticommute} and the Hurwitz-Radon-Eckmann bound} \label {app_RHE_connections} 

In \cite {Ecm}, Eckmann provided a solution to the complex version of the Hurwitz-Radon problem (and also described the solution of the original Hurwitz-Radon problem concerning real matrices). Eckmann showed that the maximum number of $n\times n$ complex matrices $A_i$ that satisfy  
\begin{enumerate}
\item $A_i A_j + A_j A_i = 0$ for all $i\neq j$,
\item $A_i^2 = -I_n$, and 
\item \label{HRE_unitary} $A_i A_i^* = I_n$
\end{enumerate}
is $2t+1$, where $t = \nu_2(n)$, i.e., the highest power of $2$ that divides $n$.  (The original Hurwitz-Radon problem asked for the maximum number of real matrices satisfying these conditions, but with Condition \ref{HRE_unitary} replaced with orthogonality: $A_i A_i^t = I_n$.)

First note that if a matrix satisfies any two of the following three conditions: 
\begin{equation}
\begin{array}{l}
\displaystyle A_i^2 = -I_n \\
\displaystyle A_i A_i^* = I_n \\
\displaystyle A_i^* = -A_i
\end{array} 
\label{HRE_equiv_cond}
\end{equation}
then it automatically satisfies the third (this is easy to see).
If we now compare the hypotheses of Theorem \ref {thm: max_Greg_anticommute} with those of the generalized Hurwitz-Radon problem, we see that Theorem  \ref {thm: max_Greg_anticommute} generalizes the Hurwitz-Radon-Ekmann bound in two ways: it does not impose any of the three conditions above in (\ref{HRE_equiv_cond}) and only considers pairwise anti commutativity, and secondly, it considers the situation where the matrices arise from the embedding of some $k$-central simple algebra, $k$ a number field, in $M_n(\C)$.  Since Theorem  \ref {thm: max_Greg_anticommute} provides a bound of $2t+3$, we find that the conditions in (\ref{HRE_equiv_cond}) drop the possible number by $2$.

%
%


\begin{thebibliography}{10}
%

\bibitem{JR}
{G.R. Jithamithra, B.S Rajan }, 
{Minimizing the Complexity of Fast Sphere Decoding of STBCs}, 
{\em { IEEE Transactions on Wireless Communications}}, vol 12, no. 12, 2013. 



%
\bibitem{VB} 
E. Viterbo, J. Boutros, ``A universal lattice decoder for fading channels,'' \emph{IEEE Trans. Inf. Theory}, vol. 45, no. 5, 1999.
%
\bibitem{BHV}
E. Biglieri, Y. Hong and E. Viterbo, ``On fast-decodable space-time block codes,''
{\em IEEE Trans. Inform. Theory}, vol. 55, no. 2, Feb 2009.
%
\bibitem{RenEtAl} T.P. Ren, Y.L. Guan, C. Yuen, and R.J. Shen, ``Fast-group-decodable space-time block code,'' Proceedings IEEE Workshop (ITW 2010), 2010.
%
\bibitem{YuenEtAl}  Chau Yuen, Yong Liang Guan, Tjeng Thiang Tjhung, ``On the Search for High-Rate Quasi-Orthogonal Space�Time Block Code,'' {\em Int. J. Wireless Inf. Network}, vol. 13, pp. 329-340, Oct. 2006.
%
\bibitem{SR} K. P. Srinath, B. S. Rajan, ``Low ML-decoding complexity, large coding gain, full-diversity STBCs for $2 \times 2$ and $4 \times 2$ MIMO systems,'' \emph{IEEE J. on Special Topics in Signal Processing: managing complexity in multi-user MIMO systems}, 2010

%
%
\bibitem{MO}
N.~Markin, F. ~Oggier, ``Iterated space-time code constructions from cyclic algebras,'' Information Theory, IEEE Transactions on, vol.59, no.9, pp.5966--5979, Sept. 2013.
%
\bibitem{OVH}
R. Vehkalahti, C. Hollanti, F. Oggier, 
``Fast-Decodable Asymmetric Space-Time Codes from Division Algebras,'' 
{\em IEEE Transactions on Information Theory}, vol. 58, no. 4, April 2012.
%
\bibitem{LO}
L. Luzzi, F. Oggier, ``A family of fast-decodable MIDO codes from crossed-product algebras over $\QQ$,''
\emph{Proc. IEEE Int. Symp. Inform. Theory}, St Petersburg, July 2011.
\bibitem{SRGAg} K. P. Srinath, B. S. Rajan, ``Generalized Silver Codes,'' \emph{IEEE Trans. Inform. Theory}, vol. 57, no. 9, Sep 2011.
%


\bibitem{NR} L. P. Natarajan, B. S. Rajan, ``Asymptotically-Good, Multigroup Decodable Space-Time Block Codes,"
{\em{IEEE Transactions on Wireless Communications}}, vol. 12, no 10, pp. 5035-5047, 2013. 
%
\bibitem{NR3}
Lakshmi Prasad Natarajan and B. Sundar Rajan, ``''Generalized Distributive Law for ML Decoding of Space-Time Block Codes,'' IEEE Trans. on Information Theory, Vol. 59, No. 5, May 2013, pp.2914-2935.
%
\bibitem{Ecm} Beno Eckmann, ``Hurwitz-Radon matrices revisited:
From effective solution of the Hurwitz matrix
equations to Bott periodicity,''
Mathematical survey lectures 1943--2004, Springer-Verlag, Berlin, 2006.
%
\bibitem{ISIT} {Gr\'egory Berhuy, Nadya Markin, B.A. Sethuraman,} `` Fast {lattice} decodability 
%
of space-time block codes,'' Proceedings of the IEEE International Symposium on Information Theory, Honolulu, Hawaii, 2014.
%
\bibitem{BOText} G. Berhuy and F. Oggier, An introduction to central simple algebras and their applications to wireless communication, Mathematical Surveys and Monographs, Amer. Math. Soc., vol. 191, 2013.
%
\bibitem{Notices} B.A. Sethuraman, ``Division algebras and wireless communications,'' \emph{Notices of the Amer. Math. Soc.}, vol. 57, pp. 1432--1439, December 2010.
%
\bibitem{SRS} B. A. Sethuraman, B. S. Rajan, and V. Shashidhar, ``Full-diversity, high-rate space-time block codes from division
algebras,'' {\em IEEE Trans. on Information Theory}, vol. 49, no. 10, pp. 2596-2616, Oct. 2003.
%
\bibitem{At} M.F. Atiyah, I.G. Macdonald, \emph{Introduction to commutative algebra. }Addison Wesley (1969)
%
\bibitem{Mi} J.S. Milne, \emph{\'{E}tale cohomology. }Princeton University Press (1980)
%
\bibitem{Kn} M.-A. Knus, \emph{Quadratic and hermitian forms over rings. Second ed. }Grund. Math. Wiss. {\bf 294} (2012)
%
\bibitem{Dem}  F.R. Demeyer, E. Ingraham, \emph{Separable algebras over commutative rings. }Lecture notes in Math. {\bf 181}, Springer, Berlin, Heidelberg, New York (1971).
%
\bibitem{Az} G. Azumaya, {\emph{On maximally central algebras. }}Nagoya Math. J., vol. {\bf 2} (1951), 119--150.
%

\end{thebibliography}
\end{document}